\documentclass[aps,pra,onecolumn,10pt,notitlepage,showpacs]{revtex4-2}
\usepackage{amsmath}  
\usepackage{amsfonts}
\usepackage{amssymb}
\usepackage{amsthm}
\usepackage{algorithm2e}
\usepackage{fullpage}
\usepackage[]{nicefrac}
\usepackage{bbm}
\usepackage{qcircuit}
\usepackage{color}
\usepackage{float}
\usepackage[normalem]{ulem}

\usepackage[final]{hyperref}
\hypersetup{
	colorlinks = true,
	allcolors = {blue},
}

\usepackage{cancel}
\usepackage{graphicx}
\usepackage{varwidth}
\usepackage{tikz,tkz-graph}
\usetikzlibrary{calc}
\usetikzlibrary{math}

\newcommand{\ignore}[1]{}
\newcommand{\ket}[1]{\left|#1\right\rangle}
\newcommand{\bra}[1]{{\left\langle#1\right|}}
\newcommand{\braket}[2]{\left\langle #1| #2 \right\rangle}
\newcommand{\ketbra}[2]{ \left| #1 \right\rangle\left\langle #2 \right|}
\newcommand{\prnt}[1]{\left( #1 \right)}
\newcommand{\prnttt}[1]{\left\{ #1 \right\}}
\newcommand{\abs}[1]{{\left| #1 \right|}}

\newcommand\norm[1]{\left\lVert#1\right\rVert}

\DeclareMathOperator{\tr}{tr}

\newcommand{\yanote}[1]{\textcolor{red}{{\bf (Yosi:}{#1}{\bf )}}}
\newcommand{\shnote}[1]{\textcolor{blue}{{\bf (SC:}{#1}{\bf )}}}

\def\th{\ensuremath{^\mathrm{th}} }

\def \eps {\varepsilon}

\newcommand{\nrm}[1]{\left\lVert #1 \right\rVert}
\newcommand{\Oo}{\ensuremath{\mathcal{O}}}

\newcommand{\pvp}{\vec{p}{\kern 0.45mm}'}
\let\oldnabla\nabla
\renewcommand{\nabla}{\oldnabla\!}

\def\Tr{\mathrm{Tr}}

\def\col{\mathrm{col}}

\newtheorem{corollary}{Corollary}
\theoremstyle{definition}
\newtheorem{definition}{Definition}
\newtheorem{lemma}{Lemma}
\newtheorem{claim}{Claim}

\begin{document}

\bibstyle{unsrt}

\title{Improved Upper Bounds for the Hitting Times of Quantum Walks} 

\author{Yosi Atia$^{1,2}$}
\email[]{gyosiat@gmail.com}
\author{Shantanav Chakraborty$^{3,4}$}  
\email[]{shchakra@iiit.ac.in}

\affiliation{$^1$ School of Computer Science and Engineering, The Hebrew University, Jerusalem,  Israel}
\affiliation{$^2$ University of Texas at Austin, Texas, USA}
\affiliation{$^3$ Centre for Quantum Science and Technology, International Institute of Information Technology Hyderabad, Telangana, India} 
\affiliation{$^4$ Centre for Security, Theory and Algorithmic Research, International Institute of Information Technology Hyderabad, Telangana, India}


\begin{abstract}
Continuous-time quantum walks have proven to be an extremely useful framework for the design of several quantum algorithms. Often, the running time of quantum algorithms in this framework is characterized by the quantum hitting time: the time required by the quantum walk to find a vertex of interest with a high probability. In this article, we provide improved upper bounds for the quantum hitting time that can be applied to several CTQW-based quantum algorithms. In particular, we apply our techniques to  the glued-trees problem, improving their hitting time upper bound by a polynomial factor: from $O(n^5)$ to $O(n^2\log n)$. Furthermore, our methods also help to exponentially improve the dependence on precision of the continuous-time quantum walk based algorithm to find a marked node on any ergodic, reversible Markov chain by Chakraborty et al. [PRA 102, 022227 (2020)].
\end{abstract}
\date{\today}
\maketitle
\section{Introduction}

Quantum walks \cite{ADZ93,Meyer1996, AAKV01, FG98, Szegedy04, Portugal13} are quantum analogues of random walks on graphs and have widespread applications in several areas of quantum information processing \cite{kempe03}. In particular, they are an universal model for quantum computation and are crucial to the design of quantum algorithms for a plethora of problems such as element distinctness \cite{Ambainis07,BCJKM13}, searching for a marked state in graphs \cite{SKW03,CG04,AKR05,Tulsi08,MNRS11,CNR18,AGJ19}, matrix product verification \cite{BS06}, triangle finding \cite{MSS07}, group commutativity \cite{MN07} and many others \cite{Kempe2002,CCDFGS03,CSV07,FGG07,montanaro15}. In fact, for some oracular problems, quantum walks exhibit an exponential speedup over their classical counterparts \cite{CCDFGS03,CSV07}. 

For several of these algorithms, the runtime strongly depends on the time required for the quantum walk on a graph to localize at a vertex inside a marked set of vertices of interest. In the context of classical random walks, the \textit{hitting time} is defined as the expected time required to hit a marked vertex. As such, the quantum hitting time characterizes the running time of the underlying quantum algorithm \cite{Kempe2002,CFG02,CCDFGS03,KB06,CNR18,MNRS11}. 

For CTQW, the Hamiltonian governing the dynamics of the walk encodes the (vertex or edge) connectivity of the underlying graph. The dynamics involves evolving some initial state according to this Hamiltonian for time $T$, followed by a measurement in the basis spanned by the vertices of the graph, thus finding a target state with a high probability. 
However, Hamiltonian evolutions are unitary, and therefore, unlike classical random walks, quantum walks never converge to a fixed state with time.
One method to address this is to evolve the system by $H$ for a random duration $t$ that is typically uniformly distributed between $0$ and $T$, followed by a measurement \cite{AAKV01, CFG02, CNR18, CLR19, CLR20}. 

In this work, we use the natural definition for quantum hitting time as the ratio between the maximal time of evolution $T$ and time-averaged probability of the walker to localize in the marked vertex. Our main observation is that CTQW is a dephasing process, or an energy measurement when the result is forgotten.  With this observation we provide improved upper bounds for the hitting time of a continuous-time quantum walk (CTQW) on a graph, and  analyse  the hitting time dependence on (i) the probability distribution according to which $t$ is distributed and (ii) the value of the maximal time of evolution $T$. 

Dephasing and dissipation have been shown to be advantageous for certain quantum processes in scenarios different from ours \cite{ALT2008, RMKL+2009, NCMO2018}. In such cases the underlying system interacts with an environment (often a thermal bath) that relaxes it fast enough to the an eigenstate having a high overlap with the state of interest. 

In this article, our insights yield two improved CTQW-based algorithms. We prove that changing the distribution of $t$ provides a polynomial improvement to the upper bound on the quantum hitting time of the glued trees problem \cite{CCDFGS03}. Additionally, we apply our techniques to the problem of finding a marked vertex on a graph, known as spatial search. We use our methods to improve exponentially, the dependence on precision of the recently developed CTQW-based spatial search algorithm \cite{CNR18}.

The input to the glued-trees problem is (i) a graph composed of two binary trees of depth $n$ each, that are glued together, and (ii) the root vertex of one of the trees, labeled $\mathtt{Entrance}$. The task is to find the root of the other binary tree labeled $\mathtt{Exit}$, where the access to the graph is by an oracle which returns the labels of the neighbors of a given vertex (See Fig.~\ref{fig:G4}).  This problem is one of the few cases where a quantum algorithm gives an exponential speedup over any classical algorithm relative to an oracle \cite{CCDFGS03}. The authors proved that the running time of this algorithm (which is the hitting time of the CTQW on the underlying graph) is in $O(n^5)$, exponentially better than any classical algorithm. In this work, we are able to improve the running time of the glued trees algorithm to $O(n^2\log n)$, providing a polynomial improvement in the algorithmic performance.

Another important application of CTQW-based algorithms is to solve the spatial search problem, where the goal is to find a marked vertex in a graph faster than classical random walks. We show that our definition of hitting time encompasses the running times of spatial search algorithms by CTQW. Childs and Goldstone introduced the first spatial search algorithm \cite{CG04} in this framework which offered a quadratic speedup over classical random walks for several graphs but fails to achieve a generic quadratic speedup \cite{CNAO16,CNR20}. Recently, in Ref.~\cite{CNR18} the authors provided a spatial search algorithm by CTQW which finds an element in a marked set of vertices on any ergodic, reversible Markov chain in the square root of the so-called extended hitting time of the corresponding classical random walk. As the extended hitting time is equal to the hitting time in the scenario where only a single vertex is marked, the algorithm offers a quadratic improvement for this case. In this article, we exponentially improve the dependence of this algorithm on precision. 

This article is organized as follows. In Sec.~\ref{sec:background}, we lay out formally the definition of the quantum hitting time and discuss prior work. In Sec.~\ref{sec:improved-hitting}, we derive improved upper bounds on the quantum hitting time, in Sec.~\ref{sec:Glued} and in Sec.~\ref{sec:search}, we apply our bounds to improve the running time of the glued trees algorithm and the spatial search algorithm, respectively. Finally, we conclude with a brief discussion and summary in Sec.~\ref{sec:discussion}. 
\section{Background}
\label{sec:background}
Consider a graph $G$ of $n$ vertices labelled $\{1,2,\cdots,n\}$ and a Hamiltonian $H$ which is a Hermitian matrix of dimension $n$ that encodes the connectivity of the underlying graph. We require that $H$ is local, i.e.\ its $(j,k)^{\mathrm{th}}$ entry is non-zero if node $j$ is adjacent to node $k$. Thus, $H$ may be proportional to the adjacency matrix of the graph or the graph Laplacian. Alternatively, CTQWs can also be defined on the edges of $G$ \cite{CNR18,CLR19}, in which case the underlying Hamiltonian preserves the local edge-connectivity of $G$. A continuous-time quantum walk on $G$ corresponds to the time-evolution of this time-independent Hamiltonian $H$, starting from some initial state $\ket{\psi_0}$.  Formally, 


\begin{definition} [CTQW] \label{def:walk}
Let $G$ be a graph of $n$ vertices and $H$ be a Hamiltonian encoding the connectivity of $G$. A CTQW on $G$ is performed by initializing the system to some state $\ket{\psi_0}$, evolving it according to $H$ for a time $t$, uniformly distributed in $[0,T]$, and finally measuring in the vertex basis of $G$. 
\end{definition} 

From this definition, depending on the choice of $T$ and $\ket{\psi_0}$, the random evolution time causes the quantum walk to converge to a fixed state in a time-averaged sense. For a classical random walk, the hitting time is the expected number of steps required to find some vertex of interest, say $y$ for the first time. However, for quantum walks, testing periodically to find if $y$ is reached is impossible, because it will turn the walk into a classical walk. To enable such a comparison, we formally define the quantum hitting time as follows:

\begin{definition}[Quantum hitting time] \label{def:HittingTime}
Let $G$ be a graph with a set of $V$ vertices such that $y$ is some vertex of interest. Furthermore let $H$ be the Hamiltonian corresponding to the CTQW on $G$. Then starting from some initial state $\ket{\psi_0}$, the hitting time of a CTQW on $G$ with respect to $y$, is defined as follows: 
\begin{equation} \label{eq:HittingTime}
\tau_G(y|\psi_0) \triangleq \min_{T>0} \frac{T}{\bar{p}_{T}(y|\psi_0)}, 
\end{equation}
wherein $\bar{p}_{T}(y|\psi_0)$ is the expected probability to find $y$ when $t$ is distributed uniformly:
\begin{equation} \label{eq:avg-prob}
\bar{p}_{T}(y|\psi_0)=\frac{1}{T}\intop_{0}^{T}\left|\langle y|e^{-iHt}|\psi_0\rangle\right|^{2}dt.
\end{equation}
\end{definition}
 
Here $\bar p_T(y|\psi_0)$ is the mean probability that a quantum walk starting from some initial state $\ket{\psi_0}$ would end in $y$ when the evolution time $t$ is uniformly distributed in $[0,T]$. By repeating the walk $1/\bar p_T(y|\psi_0)$ times, the vertex $y$ is found with a constant probability. Note that the notion of quantum hitting time can easily be generalized to the scenario where there are multiple vertices of interest.

Consider that the eigenvalues of the quantum walk Hamiltonian $H$, in descending order be $E_n\geq E_{n-1}\geq \cdots \geq E_1$ and the corresponding eigenstates be $\ket{E_n},\ket{E_{n-1}},\cdots,\ket{E_1}$ such that $H\ket{E_j}=E_j\ket{E_j}$.  Then asymptotically, 
\begin{equation}
\label{eq:limiting-distribution}
p_\infty (y|\psi_0)\triangleq\lim_{T\rightarrow\infty}\bar{p}_T(y|\psi_0)=\sum_{k}|\braket{y}{\Pi_{\mathcal V_k}|\psi_0}|^2,
\end{equation}
wherein $\Pi_{\mathcal V_k}=\ket{E_k}\bra{E_k}$. That is, the time-averaged quantum walk from Definition \ref{def:walk}, approaches the aforementioned limiting distribution, as $T\rightarrow\infty$. Thus the average probability of the quantum walk being in the vertex $y$, never exceeds Eq.~\eqref{eq:limiting-distribution}. It is natural to ask how fast a CTQW converges to this distribution. This is in essence captured by the following lemma, adapted from Ref.~\cite{CCDFGS03}.

\begin{lemma} [adapted from Lemma 1 in \cite{CCDFGS03}] \label{lem:Glued} Consider a CTQW on some graph $G$, defined by a Hamiltonian $H$,  starting from some initial state $\ket{\psi_0}$. The average probability of finding a vertex $y$ is lower  bounded as follows:
\begin{equation}
\bar{p}_T(y|\psi_0) \ge \sum_{k} 
\abs{\bra{y} \Pi_{\mathcal V_k} \ket{\psi_0}}^2 - \frac{2}{T\Delta E_{\min}},
\end{equation}
wherein $\mathcal V_k$ are eigenspaces of $H$, their respective projections are $\Pi_{\mathcal V_k}$, and $\Delta E_{\min}$ is the smallest gap between any pair of eigenvalues of the Hamiltonian. 
\end{lemma}

Note that the minimum time after which a CTQW approaches the limiting distribution $p_\infty (y|\psi_0)$ is known as the \textit{quantum mixing time} and is characterized by $\Delta E_{\min}$ \cite{CLR20}. However, $\Delta E_{\min}$ can be extremely small, leading to an extremely loose upper bound for the quantum hitting time. In the following section, we provide improvements to the upper bound on the quantum hitting time for problems where we are interested in finding a specific vertex. 

\section{Improved Upper Bounds for Hitting Times}
\label{sec:improved-hitting}
Continuous-time quantum walks, as per Definition \ref{def:walk} is a dephasing process in the eigenbasis of the underlying Hamiltonian. Equivalently, the randomized time-evolution procedure can be thought of as an energy measurement where the result of the measurement is forgotten. To see this, consider the time-averaged density matrix of the walk, which we denote by $\langle \rho(T) \rangle $ which can be written as
\begin{equation}
\label{eq:time-averaged-density-matrix}
    \langle \rho(T) \rangle = \frac{1}{T} \intop_0^T dt e^{-iHt}\rho_0 e^{iHt}=\frac{1}{T} \sum_{E_j\neq E_k} \bra{E_j}\rho_0\ket{E_k} \ketbra{E_j}{E_k} \frac{1-e^{-iT(E_j-E_k)}}{i(E_j-E_k)} + \sum_{l}\bra{E_l}\rho_0\ket{E_l} \ketbra{E_l}{E_l}
\end{equation}
When written in basis of the eigenstates of the Hamiltonian $H$, the off-diagonal elements corresponding to different energies decay with $T$. As such as $T\rightarrow\infty$, all the off-diagonal terms disappear. From the perspective of the time-energy uncertainty principle \cite{AMP02,AA17}, for  $T\approx 1/\Delta E_{\min}$, a measurement in the eigenbasis of the Hamiltonian $H$ can identify any eigenstate. Furthermore, once the measurement result is forgotten, the system is approximately a mixture of all the eigenstates. The contribution of each eigenstate to the probability of finding $y$ accumulates because there are no phases and no destructive interference. 

However, if the state we are interested in, say $y$ has a significant overlap with one or a few of the eigenstates of $H$, it suffices to choose $T$ large enough so that the system is in a mixed state between the relevant eigenstates and the rest. This way, one can improve the quantum hitting time bound, as in such cases, $T$ is significantly lower than the choice in Lemma \ref{lem:Glued}. This has been elucidated in Lemma \ref{lem:main} where we choose $T$ such that (at least) one eigenspace is dephased, instead of all the eigenstates (see proof in Appendix \ref{apndx:ProofWeakLemma}):

\begin{lemma} \label{lem:main} 
Consider a CTQW on $H$ for $t\in[0,T]$. Let $\mathcal V^*$ be an eigenspace of $H$ with energy $E^*$ such that $\Pi_\mathcal{V^*}$ is a projection on $\mathcal V^*$ and $\Delta E^*$ is the smallest gap between $E^*$ and the other eigenvalues of the Hamiltonian. Then, 
 \begin{equation}
 \label{eq:Lemmain}
\bar{p}_{T}(y|\psi_0)\ge\abs {\bra{y}\Pi_{\mathcal V^*}\ket{\psi_0}}^2  \prnt{1-\frac{4}{T\Delta E^*}}.
\end{equation}   
\end{lemma}

Note that as $T$ grows with respect to $1/\Delta E^*$, the eigenstates spanning $\mathcal V^*$ contribute to the mean probability to find $y$ after time $T$. Comparing Lemma \ref{lem:Glued} and Lemma \ref{lem:main}, we note that 
\begin{equation}
p_\infty (y|\psi_0)=\sum_{k} 
\abs{\bra{y} \Pi_{\mathcal V_k} \ket{\psi_0}}^2 
\ge
\abs{\bra{y} \Pi_{\mathcal V^*} \ket{\psi_0}}^2, 
\end{equation}
hence, for  $T\rightarrow \infty$, the asymptotic time-averaged probability to find $y$ is higher. However, Lemma \ref{lem:main} depends on $\Delta E^*$ and not on $\Delta E_{\min}$, and if for $T\approx 1/\Delta E^*$, the support of $y$ on $\mathcal V^*$ is large, an improved upper bound for the quantum hitting time is obtained. Computationally,  Lemma \ref{lem:Glued} requires knowing the smallest gap between \emph{any} two eigenvalues, while Lemma \ref{lem:main} only requires finding one eigenspace with non-negligible overlap with the initial and final states, and calculating its energy gap with respect to the rest of the spectrum. 

It is possible to generalize Lemma \ref{lem:main} further.  Consider a set $S$ of eigenstates of $H$, and let $\Pi_S$ is the projector onto the subspace spanned by $S$ such that for the problem at hand, $\nrm{\Pi_S \ket{y}}^2$ is significant (say, some constant). Ideally, in this case, it suffices to wait for a time long enough so as to dephase all the eigenstates in $S$. Formally, we require $T$ large enough so that $ \langle \rho(T) \rangle$ is a mixed state of the form
\begin{equation}
\label{eq:mixed-state}
 \langle \rho(T) \rangle = \rho_S + \rho^\perp + O(\eps),
\end{equation}
where $\rho_S$ has support only over the eigenstates in $S$, $\rho^\perp$ has support only in the orthogonal complement of $S$, and $\epsilon$ being the precision. If the eigenvalues in $S$ are separated from the rest of the eigenspace by at least $\Delta E_S$, then this is the case for 
\begin{equation}
\label{eq:time-precision}
T=O\left(\dfrac{1}{\eps}\cdot\dfrac{1}{\Delta E_S}\right).
\end{equation}

Furthermore, in order to improve the dependence on precision $\eps$, we can generalize the CTQW-scheme of Definition \ref{def:walk} slightly. Instead of evolving for a time $t\in U[0,T]$, where $U[a,b]$ is the uniform distribution in the interval $[a,b]$, we evolve the CTQW for a time chosen according to the sum of uniform random variables, known as the Irwin-Hall distribution, i.e.\ we choose $t=\sum_{j=1}^k t_j$ such that each $t_j\in U[0,T]$. Clearly for $k=1$, we get back the time-averaged evolution of Definition \ref{def:walk}. Lemma \ref{lem:better} formalizes CTQW with such time distribution (see Appendix \ref{apndx:ProofStrong} for the proof):
  

\begin{lemma} \label{lem:better}
Consider a graph $G$ and $H$ be the Hamiltonian encoding the connectivity of $G$. Consider a CTQW on $G$, starting from some state $\ket{\psi_0}$ with respect to $H$ where the evolution time is distributed according to the Irwin-Hall distribution, i.e. $t=\sum_{j=1}^k t_j$, where $t_1\dots t_k$ are i.i.d. random variables, uniformly distributed in $[0,T]$. Additionally, let $S$ be a subset of the eigenstates in $H$. Then,
\begin{equation}
\begin{split}
    \bar p_{T}(y|\psi_0) &\ge \sum_{j\in S} \abs{\braket{y}{E_j}\braket{E_j}{\psi_0}}^2 - \sqrt{3}\cdot \prnt{\frac{2}{T\Delta E_S}}^k
\end{split}
\end{equation}
wherein $\Delta E_{S}=\min_{j\in S, k}\{|E_j-E_k|\}$, namely, the minimal gap between an eigenstate in $S$ to the rest of the spectrum. Note that the maximal evolution time is $kT$.
\end{lemma} 

In the CTQW-based procedure of Definition~\ref{def:walk}, the time of evolution is a random variable  $t\sim U[0,T]$. As such, Lemma \ref{lem:better} corresponds to $k$ independent repetitions of the randomized time evolution, where we measure only at the end of the $k$ repetitions. Alternatively, taking $t=\sum_j t_j$  is equivalent to measuring the energy $k$ times with accuracy $\approx 1/T$, thus decaying the tails of the measurement distribution and inducing stronger decoherence.

The resultant mixed state is of the form
\begin{equation}
\label{eq:mixed-state-2}
 \langle \rho(T) \rangle = \rho_S + \rho^\perp + O(\delta^k),
\end{equation}
after a time $T'=O(kT)$. 

For example, if we choose the time $T\sim 4/\Delta E_S$, so that $\delta=1/2$, then by repeating the procedure of Lemma \ref{lem:better}, $k=\lceil\log_2 \sqrt{3}/\eps\rceil$ times, we obtain that
\begin{equation}
\label{eq:mixed-state-3}
 \langle \rho(T) \rangle = \rho_S + \rho^\perp + O(\eps).
\end{equation}
The total time of evolution
\begin{equation}
\label{eq:time-precision-improved}
T'=kT=O\left(\dfrac{1}{\Delta E_S}\log_2\dfrac{1}{\eps}\right).
\end{equation}

In the following sections we apply our bounds to improve the quantum hitting time of several crucial CTQW-based quantum algorithms.  
\section{Improved hitting time for the glued-trees quantum walk algorithm \label{sec:Glued}} 
In this section, apply the improved bounds on the quantum hitting time in Sec.~\ref{sec:improved-hitting} to the glued-trees quantum walk algorithm introduced in Ref.~\cite{CCDFGS03}. We begin by defining the problem followed by using relevant results from \cite{CCDFGS03}. 

\begin{definition} [Glued trees problem \cite{CCDFGS03}]
Consider two binary trees of depth $n$ such that the roots of the first and second tree are denoted \texttt{Entrance} and \texttt{Exit} respectively. Following the original notations, let $G_n'$ be the graph composed of the two trees, and of additional edges forming a random cycle, which alternates between leaves of the two trees. The   vertices of $G_n'$ are given unique unknown labels. The input to the problem is the label of the \texttt{Entrance} vertex, and the solution is the label of the \texttt{Exit} vertex.  The access to the 
$G_n'$ is by an oracle, which receives a vertex's label and returns the labels of its neighbours. 
\end{definition}


\begin{figure} [H]
\begin{center}
{
\begin{tikzpicture} [scale=0.7]

\node [ label={[xshift=-0.9cm, yshift=-0.27cm] \texttt{Entrance}},color={black}, fill=black, circle,draw,inner sep=0pt,minimum size =0.1cm] (bg0,0) at (0,0) {};

\node [ label={[xshift=11cm, yshift=-0.27cm] \texttt{Exit}},color={black}, fill=black, circle,draw,inner sep=0pt,minimum size =0.1cm] (bg0,0) at (0,0) {};

\foreach \x in {1,...,4}
{
\tikzmath{\z = 2^\x-1;}
\foreach \y in {0,1,...,\z}
{
\node [ color={black}, fill=black, circle,draw,inner sep=0pt,minimum size =0.1cm] (bg\x,\y) at (1.5*\x,{(\y-2^\x/2+1/2)*2^(3-\x)}) {};
\draw (1.5*\x,{(\y-2^\x/2+1/2)*2^(3-\x)}) -- ({1.5*(\x - 1)},{(floor(\y/2)-2^(\x-1)/2+1/2)*2^(4-\x)});
}
}
\foreach \x in {1,...,4}
{
\tikzmath{\z = 2^\x-1;}
\foreach \y in {0,1,...,\z}
{
\node [ color={black}, fill=black, circle,draw,inner sep=0pt,minimum size =0.1cm] (bg\x,\y) at ({1.5*(10-\x)},{(\y-2^\x/2+1/2)*2^(3-\x)}) {};
}
}
\foreach \x in {1,...,4}
{
\tikzmath{\z = 2^\x-1;}
\foreach \y in {0,...,\z}
{
\draw ({1.5*(10-\x)},{(\y-2^\x/2+1/2)*2^(3-\x)}) -- ({1.5*(11-\x)},{(floor(\y/2)-2^(\x-1)/2+1/2)*2^(4-\x)});
}
}

\foreach \y/\yy in {0/10,4/0,3/4,6/3,11/6,2/11,15/2,13/15,8/13,9/8,12/9,1/12,14/1,5/14,7/5,10/7}
{
\draw (9,{(\y-8+1/2)/2}) -- (6,{(\yy-8+1/2)/2});
\draw (9,{(\yy-8+1/2)/2}) -- (6,{(\y-8+1/2)/2});
}

\end{tikzpicture}}





\end{center}
\caption{\label{fig:G4}An instance of the graph $G_4'$}
\end{figure}
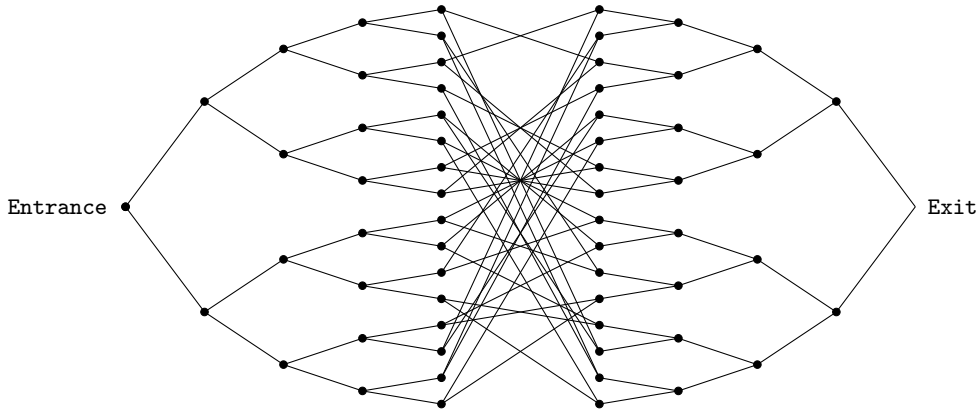

Childs et al. \cite{CCDFGS03} solved the problem using  a quantum walk on the graph $G_n'$ in polynomial time and proved an exponential speedup of their algorithm to any classical algorithm. For convenience, they have analyzed $G_{n-1}'$ instead of $G_n'$. From the symmetry of the trees and the initial state at the root, it is clear that vertices in the same depth of each tree will share the same (time-dependent) amplitude.  Following this observation, the authors defined the column states $\ket{\operatorname{col} j}$ to be an even superposition of all states with distance $j-1$ from \texttt{Entrance}. The walk is bounded to a $2n$ dimensional subspace spanned by these column states, and in it, the non-zero matrix elements of $H$ corresponding to $G_{n-1}'$ are:
\begin{equation}
\langle\operatorname{col} (j)|H| \operatorname{col}(j+1)\rangle=\langle\operatorname{col} (j+1)|H| \operatorname{col}(j)\rangle=\left\{\begin{array}{ll}{1} & {1 \leq j \leq n-1, \quad n+1 \leq j \leq 2 n-1} \\ {\sqrt{2}} & {j=n}\end{array}\right.
\end{equation}
Note that the states with $j=1,2n$ correspond to \texttt{Entrance} and \texttt{Exit} respectively. So the initial state $\ket{\psi_0}=\ket{\col\ 1}$ while the target state is $\ket{y}=\ket{\col\ 2n}$. 

The runtime of the algorithm was bounded using Lemma \ref{lem:Glued}. In order to compare the performances of Lemma \ref{lem:Glued} and our Lemma \ref{lem:main}, we repeat the essence of the spectral analysis of $H$; for full details see \cite{CCDFGS03}. 

The  eigenstates of $H$ take one of two forms \footnote{There are two additional eigenstates where $\sin$ is replaced by $\sinh$ in Eq. \ref{eq:ColH}. But $\sinh$ causes  $\abs{\bra{\mathrm {col~1}}\Pi_{\mathcal V}\ket{\mathrm{col~2n}}}$ to be exponentially small for both eigenstates, thus making it useless for our Lemma \ref{lem:main}. See Section D in \cite{CCDFGS03}.
}:
\begin{equation} \label{eq:ColH}
\begin{split}
\ket{E_p} =  \alpha_p \sum_{j=1}^n \sin p j &\ket{\operatorname{col} j} \pm \alpha_p\sum_{j=n+1}^{2n}  \sin (p(2 n+1-j)) \ket{\operatorname{col} j}
\\
&\alpha_p = \frac{1}{\sqrt{2\sum_{j=1}^n \sin^2(pj)}} ,
\end{split}
\end{equation}
with the respective eigenvalue $E_p=2\cos p$. Here, $p$ is the solution to one of the following two equations:
\begin{equation} \label{eq:psqrt2}
\frac{\sin ((n+1) p)}{\sin n p}=\pm \sqrt{2}.
\end{equation}
 A solution for  Eq. \eqref{eq:psqrt2} with the plus (minus) sign will correspond to an eigenvector with a plus (minus) sign in Eq. \eqref{eq:ColH}. 
 
 For gap calculations, \cite{CCDFGS03} proved that the solutions for $p$ corresponding to $+\sqrt{2}$ and $-\sqrt{2}$ interleave, and take  the form $p=\pi\ell/n+\delta$ and $p=\pi\ell/n-\delta$ respectively. Here, $\ell=1,2,\cdots, n-1$ and $\delta$ is a function of $n,\ell$. 
\subsection{Hitting-time bound using Lemma \ref{lem:Glued} \cite{CCDFGS03}}
In \cite{CCDFGS03}, the authors wrote an approximated solution to $p$ and proved that the minimal energy gap is around  $\ell=1$ and equals $\Delta E = \Omega(n^{-3})$.  
Recall that $f(n)=\Omega(g(n))$ iff there are constants $N_0,c>0$ such that for any $n>N_0$,  $\abs{f(n)}\ge c\abs{g(n)}$. 
By using Eq. \eqref{eq:ColH} and the Cauchy-Schwartz inequality, one can see that
\begin{equation}
\sum_k \abs{\bra{y}\Pi_{\mathcal V_k}\ket{\psi_0}}^2= \sum_{p} \abs{\braket{E_p}{\mathrm{col}~1}}^2 \abs{\braket{E_p}{\mathrm{col}~2n}}^2  =  \sum_{p} \abs{\braket{E_p}{\mathrm{col}~1}}^4  \ge \prnt{\sum_p \frac{1}{\sqrt{2n}} \abs{\braket{E_p}{\mathrm{col} ~1}}^2}^2 =\frac{1}{2n}
\end{equation}
Hence by Lemma \ref{lem:Glued}, 
\begin{equation} \label{eq:box1}
\bar{p}_T(\mathrm{col}~2n ~|~\mathrm{col}~1) \ge \sum_{p} \abs{\braket{E_p}{\mathrm{col}~1}}^2 \abs{\braket{E_p}{\mathrm{col}~n}}^2 - \frac{2}{T\Delta E} \ge \frac{1}{2n} - O\prnt{\frac{n^3}{T}}
\end{equation}
By this bound, $T$ should be of order $n^4$  to find \texttt{Exit} with probability $\approx 1/n$. Hence, the hitting time according to \cite{CCDFGS03} is $O(n^5)$.

\subsection{Hitting-time bound using Lemma \ref{lem:main}}
Lemma \ref{lem:main} requires a single eigenstate with a large gap. The following claim (see proof  Appendix \ref{apndx:claim}) is useful for that purpose \footnote{The energy gap approximation at \cite{CCDFGS03} is incorrect for $\ell=\Theta(n)$. The mistake is in the asymptotic analysis of Eq. 55. on page 15.}:   
\begin{claim} \label{clm:clm} 
The energy gap for energy levels corresponding to $p=\frac{\ell \pi}{n}-\delta$, where $\ell=\Theta(n)$, is proportional to $1/n$.
\end{claim}

Let $\mathcal V^*$  be a subspace spanned by the eigenstate corresponding to $p=\frac{\ell \pi}{n}-\delta$ where $\ell=n/2$. By Claim \ref{clm:clm},  the energy gap $\Delta E^*$ is proportional to 1/n.
. 
 The other part of Eq. \eqref{eq:Lemmain} to bound is:
\begin{equation}
\begin{split}
\abs{\langle \mathrm{col} 1|\Pi_{{\mathcal V^*}}|\mathrm{col} ~2n\rangle}&= \alpha_p^2 \sin^2 p= \alpha_p^2 \sin^2(\pi/2 +O(1/n)) > \alpha_p^2/2
\\
\alpha_p &= \frac{1}{\sqrt{2\sum_{j=1}^n \sin^2(pj)}} > \frac{1}{\sqrt{2n}}
\end{split}
\end{equation}
wherein $\alpha_p$ is the normalization factor of $\ket{E_p}$. Hence, by Lemma \ref{lem:main}:
\begin{equation} \label{eq:box2}
\bar{p}_{T}(\mathrm{col}~2n~|~\mathrm{col}~1)\ge \frac{1}{8n^2} (1- O(n/T))
\end{equation}
One can see that for $T\approx n$, the inequality in Eq. \eqref{eq:box1} is trivial, while by Eq. \eqref{eq:box2}, the probability  to find \texttt{Exit} is $\approx 1/n^2$. Hence, this improves the hitting time to $O(n^3)$. Next we prove that this can be improved further by using Lemma \ref{lem:better}.  

\subsection{Hitting-time bound using Lemma \ref{lem:better}}

Here we shall demonstrate that the procedure described in Lemma  \ref{lem:better} can be used to obtain the probability of finding the $\mathtt{Exit}$ vertex from the $\mathtt{Entrance}$ vertex for the glued trees algorithm, offering improvements over prior results. 

Recall that the $+\sqrt{2}$ solutions and $-\sqrt{2}$ solutions of Eq.~\eqref{eq:psqrt2} interleave and are satisfied by 
\begin{equation}
p=\dfrac{\ell \pi}{n}\pm \delta,
\end{equation}
where the positive sign corresponds to the $+\sqrt{2}$ solutions while the negative sign corresponds to the $-\sqrt{2}$ solutions. 

In order to use Lemma \ref{lem:better}, consider that $\lceil n/4 \rceil \leq \ell \leq \lceil 3n/4 \rceil$. If $\gamma=\ell\pi/n$, then from Appendix \ref{apndx:claim}, we have that for a $-\sqrt{2}$ solution,
\begin{align}
\tan (n\delta)&=\dfrac{\sin\gamma}{\sqrt 2 + \cos\gamma}+O(1/n).
\end{align}
For any $\ell$ in the aforementioned range, we have that $\sin \gamma \geq 1/\sqrt{2}$ and $\cos\gamma \leq 1$. Thus, for any corresponding $p$, we obtain that
\begin{equation}
\label{eq:lower-bound-delta}
\delta\geq \dfrac{\pi}{12 n}.
\end{equation} 
This lower bound for $\delta$ also holds for $+\sqrt{2}$ solutions for $p$ such that $\lceil n/4 \rceil \leq \ell \leq \lceil 3n/4 \rceil$. In order to make use of Lemma \ref{lem:better}, we need to define a set $S$ obtain $\Delta E_S$. In what follows, we prove the following claim
 
\begin{claim} \label{clm:clm-mingap} 
Let $S$ be a subset of eigenvalues $E_p$ of the Hamiltonian $H$ of the glued trees graph such that $\lceil n/4 \rceil \leq \ell \leq \lceil 3n/4 \rceil$. Then, 
\begin{equation}
\Delta E_S \geq \dfrac{\pi}{3\sqrt{2} n}.
\end{equation}
\end{claim}
\begin{proof}

We shall show that the aforementioned lower bound for $\Delta E_S$ holds for $p$ corresponding to the $-\sqrt{2}$ solutions of Eq.~\eqref{eq:psqrt2} for $\ell$ in this range. Similar result also holds for the $+\sqrt{2}$ solutions.

The absolute value of the eigenvalue gap between $E_p=2\cos p$ and the nearest eigenvalue to its left is given by
\begin{align}
\Delta E_{p,\mathrm{left}} &= 4\sin \left(\dfrac{\ell\pi}{n}-\dfrac{\pi}{2n}\right)\sin \left(\dfrac{\pi}{2n}-\delta\right)\\
                           &\geq 4\sin \left(\dfrac{\pi}{4}-\dfrac{\pi}{2n}\right)\sin \left(\dfrac{\pi}{2n}-\delta\right)\\
                           & \geq 2\sin \left(\dfrac{\pi}{2n}-\delta\right)~~~~~~~~~[\mathrm{Using }\sin\left(\pi/4-2n\right)\geq 1/2,\text{~holds for any $n\geq 6$}]\\
                           & \geq 2 \left(\dfrac{\pi}{2n}-\delta\right)\geq \dfrac{\pi}{4n},
\end{align}
where in the last line we have used the fact that $\sin x \geq x$ and $\delta\leq \frac{3\pi}{8n}$.

Similarly, the eigenvalue gap between $E_p$ and the nearest eigenvalue to its right is given by
\begin{align}
\Delta E_{p,\mathrm{right}} &=4 \sin\left(\dfrac{\ell\pi}{n}\right)\sin\delta\\
                           &=\dfrac{4\delta}{\sqrt{2}}\geq \dfrac{\pi}{3\sqrt{2}n},
\end{align} 
where in the last line we have used the lower bound for $\delta$ in Eq.~\eqref{eq:lower-bound-delta}. The claim follows by considering the minimum of $\Delta E_{p,\mathrm{left}}$ and $\Delta E_{p,\mathrm{right}}$.
\end{proof}
~\\
The improved CTQW-based algorithm to find the \texttt{Exit} vertex of the glued trees graph is summarized via Algorithm \ref{algo:improved-glued-trees}.
~\\~\\
\RestyleAlgo{boxruled}
\begin{algorithm}[ht]
  \caption{Improved glued-trees algorithm}\label{algo:improved-glued-trees}
  Choose $T=12\sqrt{2}n/\pi$ and $k=\lceil \log_2 (5\sqrt{3}n)\rceil$.
\begin{itemize}
\item Repeat the following $20n$-times.
\begin{itemize}
\item[a.~] Start from the \texttt{Entrance} vertex, i.e.\ $\ket{\psi_0}=\ket{\col\ 1}$.
\item[b.~] Evolve according to $H$ for some $t\in U[0,T]$.
\item[c.~] Repeat Step b. $k$ times.
\item[d.~] Measure in the vertex basis.
\end{itemize}
\item Output the result of the final measurement. 
\end{itemize} 
 \end{algorithm}

\begin{corollary} \label{cor:cor-alter*}
The hitting time for the glued trees quantum walk algorithm (Algorithm \ref{algo:improved-glued-trees}) is $O\left(n^2\log n\right)$, an improvement over $O(n^5)$, the previous bound proved in \cite{CCDFGS03}.
\end{corollary}
~\\
\begin{proof}
We simply apply Lemma \ref{lem:better}.
$$
 \sum_{p:E_p\in S} \abs {\braket{\col\ 2n}{E_p}\braket{E_p}{\col\ 1}}^2=\sum_{p:E_p\in S}|\alpha_p|^4\geq \dfrac{1}{4n},
$$
and by choosing $T=12\sqrt{2}n/\pi$ and $k=\lceil \log_2 (5\sqrt{3}n)\rceil$
\begin{equation}
    \sqrt{3}\prnt{\frac{2}{T\Delta E_S}}^k \leq \frac{1}{5n}.
\end{equation}
So,
\begin{equation}
   \bar{p}_{T}(\mathrm{col}~2n~|~\mathrm{col}~1) \geq \frac{1}{20 n}.
\end{equation}
Hence, the probability of the time-averaged CTQW whose time of evolution is chosen for some time $T'=kT$ according to the Irwin-Hall distribution to find the exit node is in $\Omega(1/n)$, and therefore the hitting time is in $O(n^2 \log n)$. 
\end{proof}
\section{Spatial search by continuous-time quantum walk}
\label{sec:search}

In this section,  we focus on the problem of finding a marked element on a Markov chain, known as the spatial search problem. The first CTQW-based algorithm by Childs and Goldstone \cite{CG04} could find a marked node on specific graphs of $n$ nodes such as the complete graph, hypercube and others in $O(\sqrt{n})$ time, offering a quadratic advantage over classical random walks. However, it fails to achieve any quantum speedup for other graphs such as lattices of dimension less than four. Finding out the necessary and sufficient conditions for this algorithm to be optimal for any graph, had been a long-standing open problem and considerable progress has been made recently in this regard \cite{CNAO16,CNR20}. However, the Childs and Goldstone algorithm cannot solve the spatial search problem for any ergodic, reversible Markov chain quadratically faster than its classical counterpart \cite{CNR20}. In fact recently, in Ref.~\cite{CNR18}, a new CTQW-based algorithm was developed which could find a marked node $v$ on any ergodic, reversible Markov chain $P$ in a time that is in $O\left(\sqrt{HT(P,v)}/\epsilon\right)$ with a success probability of at least $1/4-\eps$ where $HT(P,v)$ is the hitting time of a classical random walk on $P$ with respect to $v$. 

In the scenario where multiple vertices are marked, the algorithmic running time depends on a quantity known as the \textit{extended hitting time}. Given a set $M$ of marked elements, the algorithm runs in $O\left(\sqrt{HT^+(P,M)}/\epsilon\right)$, where $HT^+(P,M)$ is the extended hitting time of $P$ with respect to $M$. For instances where only a single element is marked, i.e. $|M|=1$, $HT^+(P,M)=HT(P,M)$. However, for multiple marked vertices, $HT^+(P,M)$ can be significantly greater than the hitting time.

In the framework of discrete-time quantum walks however, there exist quantum algorithms that solve the spatial search problem on any ergodic reversible Markov chain in $O\left(\sqrt{HT^+(P,M)}\log 1/\eps\right)$ time \cite{MNRS11, KMOR16}. In this section, we apply Lemma \ref{lem:multiple-evolutions}, to exponentially improve the dependence of the running time of the CTQW-based spatial search algorithm in Ref.~\cite{CNR18} on $\epsilon$ so that it has a matching running time with its discrete-time counterpart.

We begin by first briefly discuss some properties of Markov chains that we shall require for our analysis and then define a Hamiltonian corresponding to a CTQW on the edges of any Markov chain.
\\~\\
\textbf{Some basics on Markov chains:~} A Markov chain on a discrete state space $X$, such that $|X|=n$, can be described by a $n\times n$ stochastic matrix $P$ such that each entry $p_{k,l}$ of this matrix $P$ represents the probability of transitioning from state $k$ to state $l$. Any pair $(k,l)\in X \times X$, such that $p_{k,l}\neq 0$ is an edge of $P$.

Throughout this section, we shall focus our attention on ergodic, reversible Markov chains. This implies we focus our attentions on Markov chains $P$ whose eigenvalues lie between $-1$ and $1$, and  has a unique stationary state $\pi$ such that $\pi P=\pi$. The stationary state $\pi$ is a stochastic row vector and has support on all the elements of $X$. 

Let us denote it as 
\begin{equation}
\label{eqm:stationary-state-classical}
\pi=\left(\pi_1~~\pi_2~~\cdots~~\pi_n\right),
\end{equation}
such that $\sum_{j=1}^n \pi_j=1$. Also, we shall map $P\mapsto (I+P)/2$ to ensure that all the eigenvalues of $P$ lie between $0$ and $1$. This will not affect our results other than by a factor of two. 

An important quantity throughout this work is the gap between the two highest eigenvalues of $P$ (the spectral gap), which we denote by $\Delta$. 

Let $C$ be an $m\times n$ positive matrix. Then we define
\begin{equation}
\label{eq:sq-root-definition}
B=\sqrt{C},
\end{equation}
as the $m\times n$ positive matrix such that its $(i,j)^\mathrm{th}$ entry, $B_{ij}=\sqrt{C_{ij}}$. 

Following this definition, consider the discriminant matrix of $P$ which is defined as 
\begin{equation}
\label{eq:discriminant-matrix-definition}
D(P)=\sqrt{P\circ P^T},
\end{equation}
where $\circ$ indicates the Hadamard product and the $(x,y)^{\mathrm{th}}$ entry of $D(P)$ is $D_{xy}(P)=\sqrt{p_{xy}p_{yx}}$. Thus $D(P)$ is a symmetric matrix. For any ergodic, reversible Markov chain $P$, $D(P)$ is in fact \textit{similar} to $P$, i.e.\ they have the same set of eigenvalues \cite{KMOR16}. So if the eigenvalues of $P$ are ordered as $\lambda_n=1 > \lambda_{n-1}\geq\cdots\geq \lambda_1$, the spectral decomposition of $D(P)$ is
\begin{equation}
\label{eq:discriminant-matrix-spectral}
D(P)=\sum_{i=1}^{n}\lambda_i\ket{v_i}\bra{v_i},
\end{equation}
where $\ket{v_i}$ is an eigenvector of $D(P)$ with eigenvalue $\lambda_i$. Importantly, the eigenstate of $D(P)$ with eigenvalue $1$ is related to the stationary distribution of $P$, i.e.\
\begin{equation}
\label{eqmain:initial-stationary-state}
    \ket{v_n}=\sqrt{\pi^T}=\sum_{x\in X}\sqrt{\pi_x}\ket{x}.
\end{equation}
~\\
\textbf{Interpolated Markov chains}:~ Let $M\subset X$ denote the set of marked elements of the Markov chain $P$. Then given any $P$, we define $P'$ as the \textit{absorbing Markov chain} obtained from $P$ by replacing all the outgoing edges from $M$ by self-loops. Then the \textit{interpolated Markov chain} is defined as 
\begin{equation}
\label{eqmain:interpolated-mc-defintion}
P(s)=(1-s)P+sP',
\end{equation}
where $s\in[0,1]$. Clearly, $P(0)=P$ and $P(1)=P'$. Let us denote the spectral gap of $P(s)$ as $\Delta(s)$. 
Let us also define $p_M=\sum_{x\in M}\pi_x$ as the probability of obtaining a marked element in the stationary state of $P$. 

It is fair to assume that $p_M\leq 1/4$ as otherwise, a marked element can be instantaneously obtained by simply sampling from the stationary distribution itself without requiring us to run the spatial search algorithm.  

As before, the discriminant matrix of $P(s)$ is defined as 
\begin{equation}
\label{eq:discriminant-matrix-definition}
D(P(s))=\sqrt{P(s)\circ P(s)^T},
\end{equation}
where $\circ$ indicates the Hadamard product and the spectral decomposition of $D(P(s))$ is
\begin{equation}
\label{eq:discriminant-matrix-spectral_s}
D(P(s))=\sum_{i=1}^{n}\lambda_i(s)\ket{v_i(s)}\bra{v_i(s)},
\end{equation}
where $\ket{v_i(s)}$ is an eigenvector of $D(P(s))$ with eigenvalue $\lambda_i(s)$, such that $\lambda_n(s)=1 > \lambda_{n-1}(s)\geq\cdots\geq \lambda_1(s)$. Then the $1$-eigenstate of $D(P(s))$ can be expressed as 
\begin{align}
\label{eq:highest-eigenstate-discriminant-matrix}
\ket{v_n(s)}&=\ket{\pi(s)}=\sum_{x\in X}\sqrt{\pi_x(s)}\ket{x}\\
            &=\sqrt{\dfrac{(1-s)(1-p_M)}{1-s(1-p_M)}}\ket{U}+\sqrt{\dfrac{p_M}{1-s(1-p_M)}}\ket{M},
\end{align} 
where $\ket{U}$ and $\ket{M}$ are defined as
\begin{align}
\ket{U}&=\frac{1}{\sqrt{1-p_M}}\sum_{x\notin M}\sqrt{\pi_x}\ket{x}\\
\ket{M}&=\frac{1}{\sqrt{p_M}}\sum_{x\in M}\sqrt{\pi_x}\ket{x}.
\end{align}  

~\\
\textbf{Interpolated hitting time and Extended hitting time:} For any interpolated Markov chain $P(s)$, one can define a quantity known as the \textit{interpolated hitting time}~\cite{KMOR16,CNR18} as follows: 
\begin{equation}
\label{eq:interpolated-hitting-time}
HT(s)=\sum_{j=1}^{n-1}\dfrac{|\braket{v_j(s)}{U}|^2}{1-\lambda_j(s)}.
\end{equation}
The spectral gap of $P(s)$ is related to $HT(s)$ by the inequality 
\begin{equation}
\label{eq:iht-spectral-gap}
HT(s)\leq \dfrac{1}{\Delta(s)}\sum_{j=1}^{n-1}|\braket{v_j(s)}{U}|^2.
\end{equation}
For the spatial search algorithm, we shall find that the quantity of interest is the \textit{extended hitting time}. 

The \textit{extended hitting time} of $P$ with respect to a set $M$ of marked elements is defined as
\begin{equation}
\label{eq:extended-hitting-time}
HT^+(P,M)=\lim_{s\rightarrow 1} HT(s),
\end{equation}

For $|M|=1$, we have that $HT^+(P,M)=HT(P,M)$. Krovi et al. proved an explicit relationship between $HT(s)$ and $HT^+(P,M)$ \cite{KMOR16}. 

They showed that
\begin{equation}
\label{eq:interpolated-vs-extended-hitting-time}
HT(s)=\dfrac{p_M^2}{\left(1-s(1-p_M)\right)^2}HT^+(P,M).
\end{equation} 

Combining Eqs.~\eqref{eq:iht-spectral-gap} and \eqref{eq:interpolated-vs-extended-hitting-time}, we have
\begin{equation}
\label{eq:hitting-time-vs-gap}
HT^+(P,M)\leq \dfrac{1}{\Delta(s)}.\dfrac{\left(1-s(1-p_M)\right)^2}{p_M^2}\sum_{j=1}^{n-1}|\braket{v_j(s)}{U}|^2.
\end{equation} 
As we shall show subsequently, for our spatial search algorithm, we would choose a particular value of $s=s^*=1-p_M/(1-p_M)$  for which we have
\begin{equation}
\label{eq:hitting-time-vs-gap-s}
\dfrac{1}{\Delta(s^*)} \geq HT^+(P,M)/8.
\end{equation}

~\\
\textbf{Search Hamiltonian $H(s)$:~} Following Ref.~\cite{CNR18}, we define a Hamiltonian $H$, corresponding to a CTQW on the edges of $P$. Let us consider a Hilbert space $\mathcal{H}\otimes\mathcal{H}$, where $\mathcal{H}=\mathrm{span}\{\ket{x}: x\in X \}$.
Also, let $p_{xy}$ denote the $(x,y)^{\mathrm{th}}$-entry of $P$ and let $E$ be the set of edges of $P$. Define the unitary $V$ acting on $ \mathcal{H} \otimes \mathcal{H} $ such that for all $x\in X$,
\begin{equation}
\label{eq:unitary-for-hamiltonian}
V(s)\ket{x,0}=\sum_{y\in X}\sqrt{p_{xy}(s)}\ket{x,y},
\end{equation}
where the state $\ket{0}$ represents a fixed reference state in $\mathcal{H}$. Let us also define the swap operator 
\begin{equation}
 S\ket{x,y}=
\begin{cases}
\ket{y,x}, & \text{if $(x,y)\in E$} \\
\ket{x,y}, & \text{otherwise}.
\end{cases}
\end{equation}

Then the search Hamiltonian is defined as 
\begin{equation}
\label{eq:search-hamiltonian}
H=i[V(s)^\dag S V(s),\Pi_0],
\end{equation}
where $\Pi_0=I\otimes \ket{0}\bra{0}$. Crucially, the spectrum of $H$ is related to the spectrum of the discriminant matrix $D(P(s))$ and has been extensively explored in \cite{CNR18}. Here we simply state the results required for our subsequent analysis. Observe that 
\begin{equation}\label{eqmain:vn_eigenstate}
H\ket{v_n(s),0}=0,
\end{equation}
i.e.\ the $1$-eigenstate of $D(P(s))$, $\ket{v_n(s),0}$ is an eigenstate of $H(s)$ with eigenvalue $0$.

Furthermore, for $1\leq k \leq n-1$ we have the following eigenstates and eigenvalues of $H(s)$:
\begin{equation}\label{eqmain:spectrumSOham}
\ket{\Psi^\pm_k(s)}=\dfrac{\ket{v_k(s),0}\pm i\ket{v_k(s),0}^\perp}{\sqrt{2}},~~E^\pm_k=\pm \sqrt{1-\lambda_k(s)^2},
\end{equation}
where $\ket{v_k(s),0}^\perp$ is a quantum state such that $\Pi_0 \ket{v_k(s),0}^\perp = 0$.
This analysis gives us $2n-1$ out of the $n^2$ eigenvalues of $H$. It can be seen that the remaining $(n-1)^2$ eigenvalues are all $0$ and are not relevant as the algorithmic dynamics is always restricted to a subspace that is orthogonal to it.

Finally, it is important to remark that this construction of $H$ ensures that the spectral gap between the $0$ eigenvalue of $H$, which encodes the stationary state of $P$, and the rest of its eigenvalues is given by
\begin{equation}
\sqrt{1-\lambda_{n-1}(s)^2}=\Theta(\sqrt{\Delta(s)}),    
\end{equation}
i.e.\ the gap between the $0$-eigenstate of $H(s)$ and the rest is the square root of the spectral gap of $D(P(s))$. This amplification of the spectral gap is crucial for our subsequent analysis. Finally, it has also been shown that $H(s)$ corresponds to a continuous-quantum walk on the edges of $P(s)$ and we shall use this Hamiltonian in conjunction with Lemma \ref{lem:better} to improve the running time of the CTQW-based spatial search algorithm on any ergodic, reversible Markov chain.  
\\~\\
\textbf{The spatial search algorithm and its running time:~} The problem of finding an element in a marked set of vertices of a Markov chain, known as the spatial search problem, can be tackled using CTQWs. Here, we improve the running time of the spatial search algorithm of Ref.~\cite{CNR18} by exponentially improving its dependence on precision. 

Suppose that we are given any ergodic, reversible Markov chain $P$ with state space $X$ and a set $M\subset X$ of marked elements where the goal is to find some marked vertex $v\in M$.  

We shall use the CTQW-scheme of Definition \ref{def:walk} along with the Hamiltonian $H(s)$ defined in Eq.~\eqref{eq:search-hamiltonian}.
\RestyleAlgo{boxruled}
\begin{algorithm}[ht]
  \caption{Quantum spatial search by CTQW}\label{algo:search}
Consider the Hamiltonian $H(s)=i[V(s)^\dag S V(s), \Pi_0]$.  
  \begin{itemize}
  \item[1.]~Prepare the state $\ket{\pi(0),0}$.\\
  \item[2.]~For $s^*=1-p_M/(1-p_M)$, $\eps\in(0,1/4)$, $k=\lceil\log_2 \sqrt{3}/\eps\rceil$ and $T=O(\sqrt{HT^+(P,M)})$, evolve according to \\ $H(s^*)$ for a time $T'=\sum_{j=1}^k t_j$, where each $t_j$ is chosen uniformly at random between $[0,T]$.\\
  \item[3.]~Measure in the basis spanned by the state space, in the first register.
  \end{itemize}
\end{algorithm}
 The algorithm of Ref.~\cite{CNR18} is quite simple : The initial state of the algorithm is the coherent encoding of the stationary distribution of $P(0)=P$, i.e.
\begin{equation}
\ket{\psi_0}=\ket{\pi(0),0}=\sum_{j\in X}\sqrt{\pi_j}\ket{j,0}.
\end{equation}
We choose $s=s^*=1-p_M/(1-p_M)$ and evolve $\ket{\psi_0}$ for some time $t\in U[0,T]$, where $T$ is to be determined later.

Now we can apply Lemma \ref{lem:main} with $\mathcal{V}^*=\ket{v_n(s^*),0}\bra{v_n(s^*),0}$ except that now it suffices to output any marked element in $M$. So it suffices to calculate the success probability as
\begin{equation}
\label{eq:succ-prob}
p_{\mathrm{succ}}=\sum_{x\in M} \bar{p}_{T}(x|\psi_0)\geq \sum_{x\in M}\abs{\braket{x}{v_n(s^*)}\braket{v_n(s^*)}{\pi(0)}}^2.
\end{equation}
Now, 
$$
\Delta E^* = \sqrt{1-\lambda^2_{n-1}(s^*)}\leq \sqrt{2\Delta(s^*)}.
$$

Observe that for the chosen value of $s=s^*$, we have that
\begin{itemize}
\item The eigenstate $\ket{v_n(s^*)}=\dfrac{1}{\sqrt{2}}\left(\ket{U}+\ket{M}\right)$.\\
\item $\abs{\braket{v_n(s^*)}{\pi(0)}}^2\geq 1/2$.\\
\item For any $x\in M$, $\abs{\braket{x}{v_n(s^*)}}^2= \dfrac{\pi_x}{2p_M}$.\\
\end{itemize}

So from Eq.~\eqref{eq:Lemmain}, for any $x\in M$ we have
\begin{equation}
\bar{p}_{T}(x|\psi_0)\ge\abs {\braket{x}{v_n(s^*)}\braket{v_n(s^*)}{\pi(0)}}^2  \prnt{1-\frac{4}{T\Delta E^*}} \geq \dfrac{\pi_x}{4p_M}\prnt{1-\frac{4}{T\Delta E^*}}.
\end{equation}

The success probability
\begin{equation}
p_{\mathrm{succ}}=\sum_{x\in M} \bar{p}_{T}(x|\psi_0)\geq \dfrac{1}{4}\prnt{1-\frac{4}{T\Delta E^*}}
\end{equation}

So for any $\eps \in (0,1)$ we ensure that the success probability is 
\begin{equation}
p_{\mathrm{succ}}\geq \dfrac{1}{4}\prnt{1-\eps},
\end{equation}
by choosing any 
\begin{align}
T&\geq \dfrac{1}{\eps}\cdot\dfrac{4}{\Delta E^*}\\
 &\geq \dfrac{1}{\eps}\cdot\dfrac{2\sqrt{2}}{\sqrt{\Delta(s^*)}}\\
 &\geq  \dfrac{1}{\eps}\cdot\sqrt{\dfrac{HT^+(P,M)}{8}},
\end{align}
where in the last line we have used the inequality in Eq.~\eqref{eq:hitting-time-vs-gap-s}. Thus we recover the running time of the spatial search algorithm of \cite{CNR18}. 

In order to improve the dependence on precision, consider Algorithm \ref{algo:search}. Formally, we prove that
~\\
\begin{lemma}
\label{lem:search}
\label{thm-main:search-phase-randomization}
For any ergodic, reversible Markov chain $P$ with a set $M$ of marked elements, Algorithm \ref{algo:search} has a success probability of 
$$p_{\mathrm{succ}}\geq 1/4-\eps,$$
for
$$T= O\left(\sqrt{HT^+(P,M)}\log \dfrac{1}{\eps}\right),$$ 
where $HT^+(P,M)$ is the extended hitting time of a random walk on $P$ with respect to $M$.  
\end{lemma}
\begin{proof}
We apply Lemma \ref{lem:better}, starting from $\ket{\pi(0),0}$ and by choosing $S=\ket{v_n(s^*),0}\bra{v_n(s^*),0}$. We obtain that
for any $x\in M$,
\begin{align}
    p_{\mathrm{succ}}&\geq \sum_{x\in M}\abs{\braket{x}{v_n(s^*)}\braket{v_n(s^*)}{\pi(0)}}^2 - \sqrt{3}\cdot \prnt{\frac{2}{T\sqrt{\Delta(s^*)}}}^k\\
                     & \geq \dfrac{1}{4}-\eps,
\end{align}
for any 
\begin{align}
T \geq \dfrac{4}{\sqrt{\Delta(s^*)}}\geq \sqrt{2 HT^+(P,M)},
\end{align} 
and by choosing $k=\lceil\log_2 \left(\frac{\sqrt{3}}{\eps}\right)\rceil$, where $\eps\in (0,1/4)$. Thus the overall running time 
\begin{equation}
T'= kT=O\left(\sqrt{HT^+(P,M)}\log \dfrac{1}{\eps}\right),
\end{equation}
thereby improving the dependence on $\eps$ exponentially.
\end{proof}

As mentioned before, the extended hitting time is equal to the hitting time when $|M|=1$, thus giving a full quadratic speedup for the spatial search problem over classical random walks. However when multiple vertices are marked, the extended hitting time can be larger than the hitting time. As a result, unlike in the discrete-time quantum walk framework \cite{AGJ19}, the problem of whether a full quadratic speedup is possible in the CTQW framework, even in the case of multiple marked vertices had been open. Recently, Apers et al. have managed to close this problem, the details of which shall appear elsewhere \cite{ACNR21}.

\section{Discussion}
\label{sec:discussion}
In this article, we provided a more general definition of the hitting time of continuous-time quantum walks and have elucidated strategies by which hitting times of continuous-time quantum walk based algorithms can be improved. Furthermore, we have applied our results to improve the running times of two important CTQW-based algorithms namely the glued-trees algorithm and the spatial search algorithm.

For the glued-trees algorithm, we have improved its running time from $O(n^5)$ in Ref.~\cite{CCDFGS03} to $O(n^2\log n)$. The source of the improvement is our bound's dependence on the gap between some eigenspace $S$ of the Hamiltonian defining the quantum walk and the rest. We have shown that $S$ is comprised of $n/2$-eigenstates, each of which add to the probability of finding the $\mathtt{Exit}$ vertex. As a result,  the quantum walk needs near linear time $T\approx n\log n$ to find the \texttt{Exit} with probability $\approx 1/n$. 

Of course, this also implies that by substituting classical repetitions with the amplitude amplification procedure \cite{BHMT00} (a quantum algorithm in the circuit model), the running time of the Glued trees algorithm can be improved to $O\left(n^{3/2}\log n\right)$. However, since we are working in framework of continuous-time quantum walks, an analog model of quantum computation, we assumed that we only have access to the continuous time-evolution under the quantum walk Hamiltonian, leading to a $O(n^2\log n)$ running time. In this regard, one direction of future research would be to explore the possibility of improving the complexity of the Glued trees algorithm to linear in $n$ either by exactly solving the underlying Schr\"{o}dinger Equation or by attaching semi-infinite pathways to the vertices of the Glued-trees graph and scattering wave packets off of it. The ballistic spread of wavepackets would underpin the possibility of an improved running time, using techniques that are crucial to demonstrate that quantum walks are universal for quantum computing \cite{childs2013universal}.

For the spatial search algorithm, our methods help find an element from a marked set $M$ in any ergodic, reversible Markov chain $P$ with success probability $1/4-\eps$ in time $T=O(\sqrt{HT^+(P,M)}\log 1/\eps)$ time, for some $\eps\in (0,1/4)$, whereas previously the algorithm of \cite{CNR18} required a time of $O(\sqrt{HT^+(P,M)}/\eps)$, thereby improving the dependence on precision exponentially. The improvement in the running time stems from evolving the search Hamiltonian for a time chosen from the Irwin-Hall distribution (sum of uniform random variables), as opposed to \cite{CNR18}, where the CTQW evolution time is chosen uniformly at random from some interval.

Our work opens up several interesting questions. One natural question to ask is what dependence does the   choice of distribution of the random evolution time of the CTQW have on the quantum hitting time? Furthermore, for a given CTQW algorithm, what is the distribution that minimizes the quantum hitting time?  Our bounds on the quantum hitting times are quite general and can be applied to other CTQW-based algorithms. For example, they can help improve the algorithmic performance of the quantum algorithm to sample from the stationary state of any ergodic, reversible Markov chain of Ref.~\cite{CLR19}. In particular, the algorithm therein has a considerable overhead in terms of the number of ancilla qubits which can, in principle be overcome using the techniques we presented here. 

Some of the techniques presented in the article can have broader applications that go beyond continuous-time quantum walks. As mentioned before, randomized time-evolution introduces dephasing in the eigenbasis of the Hamiltonian, thereby decoupling certain eigenstates from the rest. This results in a mixed state spanned by the relevant eigenstates. As such, this provides a continuous-time procedure that can be applied to state-preparation problems, which are ubiquitous throughout quantum computation. For example, numerical evidence suggests that continuous-time quantum walks can be used to find the ground states of spin glasses \cite{CCMK2019} with a super-quadratic scaling in the running time. It would be interesting to explore whether our techniques can provide analytical insights into the observations therein.  
\begin{acknowledgments}
Y.A. and S.C. thank Andrew Childs for helpful discussions. Y.A. thanks Dorit Aharonov and Patrick Rall for useful discussions. Y.A. acknowledges support from  ERC grant number 280157,  and from Simons foundation grants 385590 and 385586. S.C. thanks Tanima Karmakar for inspiring discussions and acknowledges support from IIIT Hyderabad. 
\end{acknowledgments}

\begin{appendix}
\section{Proof of Lemma \ref{lem:main}\label{apndx:ProofWeakLemma}}
\setcounter{lemma}{1}
Recall that
\begin{lemma} 
Consider a CTQW in $H$ with $t\in[0,T]$. Let $\mathcal V^*$ be an eigenspace of $H$ with energy $E^*$. Then, 
 \begin{equation}
\bar{p}_{T}(y|\psi_0)\ge\abs {\bra{y}\Pi_{\mathcal V^*}\ket{\psi_0}}^2  \prnt{1-\frac{4}{T\Delta E^*}}
\end{equation}  
wherein $\Pi_\mathcal{V^*}$ is a projection on $\mathcal V^*$; and $\Delta E^*$ is the smallest gap between $E^*$ and the other eigenvalues of the Hamiltonian. 
\end{lemma}



\begin{proof} For the proof, we first isolate one eigenspace $\mathcal V^*$ with energy $E^*$ from the rest of the eigenstates.
\begin{equation}
\begin{split}
\bar{p}_{T}(y|\psi_0)&=\frac{1}{T}\intop_0^T dt \abs {\bra{y}e^{-iHt}\ket{\psi_0}}^2
=
\frac{1}{T}\intop_0^T dt \abs {e^{-iE^* t}\bra{y}\Pi_{\mathcal V^*}\ket{\psi_0}+\sum_{k:\mathcal V_k \perp \mathcal V^*}e^{-iE_kt}\bra{y}\Pi_{\mathcal V_k}\ket{\psi_0} }^2 
=
\\
&\frac{1}{T}\intop_0^T dt \abs {\bra{y}\Pi_{\mathcal V^*}\ket{\psi_0}+\sum_{k:\mathcal V_k \perp \mathcal V^*}e^{-i(E_k-E^*)t}\bra{y}\Pi_{\mathcal V_k}\ket{\psi_0} }^2  
\end{split}
\end{equation}

The time-dependent sum coherently destroys $\bra{y}\Pi_{\mathcal{V}^*}\ket{\psi_0}$ at $t=0$, but cannot sustain the destructive interference for long. Using the inequality $\abs{a+b}^2 \ge \abs{a}^2 + 2 \mathrm {Re}(ab^*)$, we get
\begin{equation} \label{eq:AB}
\begin{split}
\bar{p}_{T}(y|\psi_0) &=
... 
\ge
\frac{1}{T}\intop_0^T dt \abs {\bra{y}\Pi_{\mathcal V^*}\ket{\psi_0}}^2 +\frac{2}{T}\mathrm {Re} \intop_0^T dt \prnt{\bra{y}\Pi_{\mathcal V^*}\ket{\psi_0}\sum_{k:\mathcal V_k \perp \mathcal V^*}e^{i(E_k-E^*)t}\bra{x}\Pi_{\mathcal V_k}\ket{y}} 
\\
= &
 \abs {\bra{y}\Pi_{\mathcal V^*}\ket{\psi_0}}^2 +\frac{2}{T}\mathrm {Re} \sum_{k:\mathcal V_k \perp \mathcal V^*} {\bra{y}\Pi_{\mathcal V^*}\ket{\psi_0}\bra{\psi_0}\Pi_{\mathcal V_k}\ket{y}\frac{e^{i(E_k-E^*)T}-1}{i(E_k-E^*)}}
\end{split}
\end{equation}
 We bound the second term: 
\begin{equation}
    \begin{split}
    \bar{p}_{T}(y|\psi_0) 
&=
... 
\ge
\abs {\bra{y}\Pi_{\mathcal V^*}\ket{\psi_0}}^2 -\frac{2}{T} \abs{ \sum_{k:\mathcal V_k \perp \mathcal V^*} {\bra{y}\Pi_{\mathcal V^*}\ket{\psi_0}\bra{x}\Pi_{\mathcal V_k}\ket{y}\frac{e^{i(E_k-E^*)T}-1}{i(E_k-E^*)}}}
\\
\ge&
\abs {\bra{y}\Pi_{\mathcal V^*}\ket{\psi_0}}^2 - \frac{2}{T} \abs{ \sum_{k:\mathcal V_k \perp \mathcal V^*} {\bra{y}\Pi_{\mathcal V^*}\ket{\psi_0}\bra{\psi_0}\Pi_{\mathcal V_k}\ket{y}} }
\cdot \frac{2}{\Delta E^*}.
\end{split}
\end{equation}
From orthogonality,
$$
\braket{y}{\psi_0}=\bra{y}\Pi_{\mathcal V^*}\ket{\psi_0} +\sum_{k:\mathcal V_k \perp \mathcal V^*} \bra{y}\Pi_{\mathcal V_k}\ket{\psi_0}=0
$$
Hence, 
\begin{equation}
    \begin{split}
    \bar{p}_{T}(y|\psi_0) 
&=
... 
\ge
\abs {\bra{y}\Pi_{\mathcal V^*}\ket{\psi_0}}^2 - \frac{4}{T\Delta E^*} \abs{ \sum_{k:\mathcal V_k \perp \mathcal V^*} {\bra{y}\Pi_{\mathcal V^*}\ket{\psi_0}\bra{\psi_0}\Pi_{\mathcal V_k}\ket{y}} }= \abs {\bra{y}\Pi_{\mathcal V^*}\ket{\psi_0}}^2  \prnt{1-\frac{4}{T\Delta E^*}}
\end{split}
\end{equation}

\end{proof}

\setcounter{lemma}{3}

\section{Proof of Lemma \ref{lem:better}\label{apndx:ProofStrong}}

The characteristic function provides a useful tool for our analysis. Consider a random variable $X\in \mathbb{R}$ from some continuous probability distribution such that its probability density function is defined as $f_X$. Then the characteristic function of $X$ is defined as
\begin{equation}
\Phi_X(r)=\mathbb{E}\left[e^{irX}\right]=\int_{\mathbb{R}}e^{irX} f_X\ dx
\end{equation}   

For example if $X\in U[0,T]$, where $U[0,T]$ is the uniform distribution defined in the interval $[0,T]$, then
\begin{equation}
\Phi_X(r)=\dfrac{e^{iTr}-1}{iTr}.
\end{equation}
Observe that the transformation induced by a single randomized time evolution is of the following form
\begin{align}
\label{eq:maps-characteristic-function-uniform-1}
\ket{E_k}\bra{E_k}&\mapsto \ket{E_k}\bra{E_k}\\
\label{eq:maps-characteristic-function-uniform-2}
\ket{E_k}\bra{E_j}&\mapsto \dfrac{1}{T}\intop_{0}^{T} dt \ e^{-i(E_k-E_j)t} \ket{E_k}\bra{E_j}\ dt=\Phi_X(E_j-E_k) \ket{E_k}\bra{E_j}
\end{align} 
Suppose that now  the Hamiltonian $H$ is evolved (without a measurement) first for time $t'=t_1+t_2+\cdots+t_k$, such that each $t_i\in U[0,T]$. Then, the maximal total time of evolution is $T'=kT$, and the overall evolution time is a random variable from a sum of uniform random variables. It is well known that $T'$  follows the Irwin-Hall distribution \cite{JKB94} and the characteristic function of $T'$ is given by
\begin{equation}
\Phi_{T'}(r)=\left(\dfrac{e^{irT}-1}{iTr}\right)^k.
\end{equation}
As in Eq.~\eqref{eq:maps-characteristic-function-uniform-1} and in Eq.~\eqref{eq:maps-characteristic-function-uniform-2} we have,
\begin{align}
\label{eq:maps-characteristic-function-sum-uniform}
\ket{E_k}\bra{E_k}&\mapsto \ket{E_k}\bra{E_k}\\
\ket{E_k}\bra{E_l}&\mapsto \Phi_{T'}(E_l-E_k) \ket{E_k}\bra{E_l}=\left(\frac{e^{i(E_k-E_l)T}-1}{i(E_k-E_l)T}\right)^k\ket{E_k}\bra{E_l}.
\end{align}

Let us denote the Frobenius norm of an operator $A$ as $||A||_F$. Then in order to prove Lemma \ref{lem:better}, we first show that starting from some $\rho_0=\ket{\psi_0}\bra{\psi_0}$, where $\ket{\psi_0}=\sum_{j=1}^n c_j\ket{E_j}$, the time-averaged density matrix $\langle \rho(T')\rangle$ becomes close (in Frobenius norm) to the following density matrix
\begin{equation}
\label{eq:dephased-density-matrix-repeated-randomization}
\bar{\rho'}=\sum_{k:E_k\in \mathcal{V_S^*} } |c_k|^2 \ket{E_k}\bra{E_k} + \rho'^{\perp},
\end{equation}
where
$$
\rho'^{\perp}=\sum_{\substack{k:E_k\notin S\\l:E_l\notin S}}  c_k c_l^*\  \Phi_{T'}(E_l-E_k)\ket{E_k}\bra{E_l}.
$$
This has been demonstrated in the following lemma.
~\\
\begin{lemma}
\label{lem:multiple-evolutions}
Let $T'=t_1+\dots+t_k$ where $t_1,\dots,t_k$ are i.i.d.  distributed uniformly in $[0,T]$. Let $\bar{\rho'}$ be defined as in Eq.~\eqref{eq:dephased-density-matrix-repeated-randomization}. Then the CTQW with the distribution $T'$ results in a state $\langle \rho(T') \rangle$ such that
$$
\nrm{\langle \rho(T') \rangle - \bar{\rho'}}_F\leq \sqrt{3} \cdot \prnt{\frac{2}{T\Delta E_S}}^k.
$$

\end{lemma}
~\\
\begin{proof}
The time evolved state according to this new distribution is
\begin{equation}
\langle \rho(T') \rangle = \bar{\rho'}+\sum_{\substack{E_k, E_l\in S \\ E_k\neq E_l}} c_kc_l^*\left(\frac{e^{i(E_k-E_l)T}-1}{i(E_k-E_l)T}\right)^k\ket{E_k}\bra{E_l}
+
\prnt{
\sum_{\substack{E_p\notin S \\ E_q\in S}} c_pc_q^*\left(\frac{e^{i(E_p-E_q)T}-1}{i(E_p-E_q)T}\right)^k\ket{E_p}\bra{E_q}+h.c}.
\end{equation}

Then,
\begin{align}
\nrm{\langle \rho(T') \rangle - \bar{\rho'}}_F&= \nrm{\sum_{\substack{E_k, E_l\in S \\ E_k\neq E_l}} c_kc_l^*\left(\frac{e^{i(E_k-E_l)T}-1}{i(E_k-E_l)T}\right)^k\ket{E_k}\bra{E_l}
+
\prnt{
\sum_{\substack{E_p\notin S \\ E_q\in S}} c_pc_q^*\left(\frac{e^{i(E_p-E_q)T}-1}{i(E_p-E_q)T}\right)^k\ket{E_p}\bra{E_q}+h.c.}}_F\\
                                                \label{eq:mixed-state-better-bound}
                  &= \sqrt{\sum_{\substack{E_k, E_l\in S\\ E_k\neq E_l}} \abs{c_kc_l}^2\abs{\frac{e^{i(E_k-E_l)T}-1}{(E_k-E_l)T}}^{2k} + 2\sum_{\substack{E_p\notin S \\ E_q\in S}} \abs{c_pc_q}^2\abs{\frac{e^{i(E_p-E_q)T}-1}{(E_p-E_q)T}}^{2k}}
\\
&\le \sqrt{3} \cdot \prnt{\frac{2}{T\Delta E_S}}^k
\end{align}


\end{proof}


Lemma \ref{lem:multiple-evolutions} shows that the for large enough $T$, the density matrix becomes a mixture of eigenstates in $S$ and some residue matrix $\rho'^\perp$ with no support on these eigenstates. This gives an improved bound to the hitting time, because now all the eigenstates in $S$ contribute to the probability to find $y$ and consequently, we directly obtain that
\begin{equation}
\begin{split}
\bar p_T(y|\psi_0) &= \tr\prnt{\ketbra{y}{y} \langle \rho (T')\rangle} \ge \tr\prnt{\ketbra{y}{y} \bar\rho'} -\tr(\ketbra{y}{y})\cdot  \sqrt{3} \cdot \prnt{\frac{2}{T\Delta E_S}}^k
\end{split}
\end{equation}
The proof of Lemma \ref{lem:better} concludes by taking only the population of $\bar\rho'$ corresponding to $S$.

\section{\label{apndx:claim} Proof of Claim \ref{clm:clm}}
\addtocounter{claim}{-1}
\begin{claim} 
The energy gap for energy levels corresponding to $p=\frac{\ell \pi}{n}-\delta$, where $\ell=\Theta(n)$, is proportional to $1/n$.
\end{claim}
\begin{proof}
Recall that the eigenvalues take the form $E_p=2\cos p$ where:
\begin{equation} \label{eq:p_constraint}
    \frac{\sin((n+1)p)}{\sin np}=\pm\sqrt 2.
\end{equation}
Furthermore, the eigenvalues which correspond to the positive RHS in Eq. \ref{eq:p_constraint} interleave with those with the negative RHS \cite{CCDFGS03}. Following \cite{CCDFGS03}, we substitute $p$ with $\frac{\ell\pi}{n}\pm \delta$  for the RHS being $\pm \sqrt 2$, wherein $\delta>0$.  From the interleaving property, $\delta \le \frac{\pi}{n}$. 

First we solve for  $-\sqrt 2$, by finding the smallest $\delta >0$ satisfying:
\begin{equation} \label{eq:sqrt2}
    -\sqrt 2 \sin (n\delta) = \sin (n\delta -\ell \pi/n +\delta)
\end{equation}
We define $\gamma =\ell \pi/n \in (0,\pi)$, which is a constant for $\ell=\Theta(n)$.
\begin{equation} \label{eq:Trigo}
    \begin{split}
        -\sqrt 2 \sin (n\delta)&=\sin (n\delta) \cos (\delta-\gamma)+\cos (n\delta) \sin (\delta - \gamma) 
        \\
        \sin (n\delta)(-\sqrt 2 -\cos(\delta-\gamma))&=\cos(n\delta)\sin(\delta-\gamma)
        \\
        \tan (n\delta)&=\frac{\sin(\gamma-\delta)}{\sqrt 2 + \cos(\delta-\gamma)} 
    \end{split}
\end{equation}
The RHS of the the last line in Eq. \ref{eq:Trigo} is positive because  $\delta=O(1/n)$ and $\gamma\in(0,\pi)$ is a constant. Hence, $n\delta\le  \frac{\pi}{2}$.  Bounding the RHS we get
\begin{equation}
\begin{split}
      \tan (n\delta)&=\frac{\sin(\gamma-\delta)}{\sqrt 2 + \cos(\delta-\gamma)} \le \frac{1}{\sqrt{2}-1}.
\\
      \tan (n\delta)&=\frac{\sin(\gamma-\delta)}{\sqrt 2 + \cos(\delta-\gamma)} = \frac{\sin(\gamma)}{\sqrt 2 + \cos(\gamma)}+O(1/n)
\end{split}
\end{equation}
Hence, for $\ell=\Theta(n)$, $n\delta$ is a positive constant smaller than $3\pi/8$.

Similarly for a positive $\sqrt 2$ at Eq. \ref{eq:sqrt2}, and $p=\frac{\ell \pi}{n}+\delta$ we get:
\begin{equation}
\begin{split}
    \sqrt 2 \sin (n\delta)&=\sin(n\delta + \ell\pi/n+\delta)
    \\
    \sqrt 2 \sin(n\delta)&=\sin(n\delta)\cos(\gamma+\delta)+\cos(n\delta)\sin(\gamma+\delta)
    \\
    \sin(n\delta)(\sqrt 2 -\cos(\gamma+\delta))&=\cos(n\delta)\sin(\gamma+\delta)
    \\
    \tan(n\delta)&=\frac{\sin(\gamma+\delta)}{\sqrt 2 - \cos(\gamma+\delta)}, 
\end{split}
\end{equation}
and the rest of the analysis is the same.

In conclusion, we get that for $+\sqrt 2$ solutions, \begin{equation}    
\Theta(1/n)= p- \ell\pi/n \le  3\pi/8n,
\end{equation}
and for $-\sqrt 2$ solutions,
\begin{equation}    
\Theta(1/n)= \ell\pi/n- p\le 3\pi/8n.
\end{equation}
Hence, for $\ell=\Theta(n)$, the gap between $p$ solutions is $\Theta(1/n)$, and
\begin{equation}
    \Delta E = 2 \sin p \Delta p = \Theta(1/n).
\end{equation}


\end{proof}
\end{appendix}

\bibliography{bib}
\bibliographystyle{unsrt}
\end{document}